\documentclass[10pt]{article}

\usepackage[margin=2.5cm, includefoot, footskip=30pt]{geometry}

\usepackage{amssymb}
\usepackage{amsmath}
\usepackage{amsthm}
\usepackage{booktabs}
\usepackage{caption}
\usepackage{epsfig}
\usepackage{graphicx}
\usepackage{hyperref}
\usepackage{psfrag}
\usepackage{subcaption}
\usepackage{ulem}
\usepackage{standalone}
\usepackage{tikz}
\usetikzlibrary{calc}
\usepackage{breakcites}
\usepackage{setspace}
\usepackage{lipsum}
\usepackage{breqn}

\newcommand{\be}{\begin{eqnarray}}
\newcommand{\ee}{\end{eqnarray}}
\usepackage{tabularx} 

\newtheorem{theorem}{Theorem}

\usepackage{lineno}

\title{An Evolutionary Game Theoretic Model of Rhino Horn Devaluation}
\author{Nikoleta E. Glynatsi, Vincent Knight, Tamsin E. Lee}
\date{}

\begin{document}

\maketitle

\indent Poachers' interactions; Rhinoceros; Wildlife

\bigskip

\textbf{Declaration of interest}: None.

\newpage
\begin{abstract}

Rhino populations are at a critical level due to the demand for rhino horn and
the subsequent poaching. Wildlife managers attempt to secure rhinos with
approaches to devalue the horn, the most common of which is dehorning. Game theory
has been used to examine the interaction of poachers and wildlife managers where
a manager can either `dehorn' their rhinos or leave the horn attached and poachers
may behave `selectively' or `indiscriminately'. The approach described in this paper
builds on this previous
work and investigates the interactions between the poachers. We build an evolutionary
game theoretic model and determine which strategy is preferred by a poacher in various
different populations of poachers. The purpose of this work is to discover whether
conditions which encourage the poachers to behave selectively exist, that is,
they only kill those rhinos with full horns.

The analytical
results show that full devaluation of all rhinos will likely lead to
indiscriminate poaching. In turn it shows that devaluing of rhinos can only be
effective when implemented along with a strong disincentive framework.
This paper aims to
contribute to the necessary research required for informed discussion about the
lively debate on legalising rhino horn trade.

\end{abstract}

\section{Introduction}\label{section:introduction}

Rhino populations now persist largely in protected areas or on private land, and
require intensive protection~\cite{Ferreira2014} because the demand for rhino
horn continues to pose a serious threat~\cite{Amin2006}. The illegal trade in rhino
horn supports aggressive poaching syndicates and a black market~\cite{Nowell1992, Warchol2003}.
This lucrative market entices people to invest their time and energy to gain a
`windfall' in the form of a rhino horn, through the poaching of rhinos.

Standard economic theory predicts extinction through poaching alone is unlikely due to
escalating costs as the number of remaining species approaches zero~\cite{courchamp2006rarity}.
However, the rarity of rhino horn makes it a luxury good, or financial investment
for the wealthy~\cite{gao2016rhino}, and thus the increased cost and risk to poach
does not increase as rapidly as the increased gain - the anthropogenic
Allee effect~\cite{branch2013opportunistic, courchamp2006rarity}. However, the
anthropogenic Allee effect was recently revisited~\cite{holden2017high} to highlight
that the relationship is even more complex and pessimistic. The value of rhino horn
can inflate, even with a large population size, due to an increase in the cost (i.e. risk)
to poach. Therefore measures to protect rhino horn may actually be increasing the gain
to poachers. It is not clear whether this relationship has contributed to the escalation
in rhino poaching over recent years. Nonetheless, it is clear that the future
existence of rhinos is endangered because of poaching~\cite{Duan2013,Smith1993}.
This rationale leads to debate about legalising rhino horn trade, which in turn may
reduce demand. In~\cite{Duan2013} the authors suggest meeting the demand for rhino horn
through a legal market by farming the rhino horn from live rhinos. In fact recently
the actual quantity of horn that could be farmed was estimated by \cite{Taylor2017}.
However~\cite{crookes2016trading} argues that because the demand for horn is so high,
legalising trade may lead to practices that maximise profit, but are not suitable for
sustainable rhino populations, and thus rhinos may be `traded on extinction'.
Preventing poaching covers in-country and global issues, and thus legalising rhino
horn trade is a controversial and active conversation, which is not limited to rhinos -
\cite{Harvey2017} considered ivory and stated that by enforcing a domestic ivory
trade ban we can reduce the market's demand.

As it stands, for wildlife managers law enforcement is often one of the main methods
to deter poachers. Rhino conservation has seen increased militarisation with `boots
on the ground' and `eyes in the sky'~\cite{duffy_st}. An alternative method is to
devalue the horn itself, one of the main  methods being the removal so that only
a stub is left. The potential impact of various policies are nicely summarised
in~\cite{crookes2016categorisation}, where de-horning is noted to be promising
for `in-country intervention'.
The first attempt at large-scale rhino dehorning as an anti-poaching
measure was in Damaraland, Namibia, in 1989~\cite{Milner1992}. Other methods
of devaluing the horn that have been suggested include
the insertion of poisons, dyes or GPS trackers~\cite{Gill2010, Smith1993}. However,
like dehorning, they cannot remove all the potential gain from an intact horn
(poison and dyes fade or GPS trackers can be removed and have been found to affect
only a small proportion of the horn).
In~\cite{Milner1992, milner1999many} they found the
optimum proportion to dehorn using mean horn length as a measure of the
proportion of rhinos dehorned. They showed, with realistic parameter values,
that the optimal strategy is to dehorn as many rhinos as possible.
A manager does not need to choose between law enforcement or devaluing, but
perhaps adopt a combination of the two; especially given that devaluing rhinos
comes at a cost to the manager, and the process comes with a risk to the rhinos.

A recent paper modelled the interaction between a rhino manager and poachers
using game theory~\cite{Lee}. The authors consider a working year of a single
rhino manager. A manager is assumed to have standard yearly resources which
can be allocated on devaluing a proportion of their rhinos or spent on security.
It is assumed that all rhinos initially have intact horns. Poachers may either only
kill rhinos with full horns, `selective poachers', or kill all rhinos they encounter,
`indiscriminate poachers'. This strategy may be preferred to avoid tracking a
devalued rhino again, and/or to gain the value from the partial horn. If all rhinos
are left by the rhino manager with their intact horns, it does not pay poachers to
be selective so they will chose to be indiscriminate since being selective incurs
an additional cost to discern the status of the rhino. Conversely, if all poachers are
selective, it pays rhino managers to invest in devaluing their rhinos.
This dynamic is represented in Fig.~\ref{fig:RhinoPic}.
Assuming poachers and managers will always behave so as to maximise their payoff,
there are two equilibriums: either all rhinos are devalued and all poachers are selective;
or all horns are intact and all poachers are indiscriminate. Essentially, either the managers
win, the top left quadrant of Fig.~\ref{fig:RhinoPic}, or the poachers win, the bottom
right quadrant of Fig.~\ref{fig:RhinoPic}. The paper~\cite{Lee} concludes that poachers will
always choose to behave indiscriminately, and thus the game settles to the top
left quadrant, i.e., the poachers win.

\begin{figure}[!htbp]
    \centering
    \includegraphics[scale=0.2]{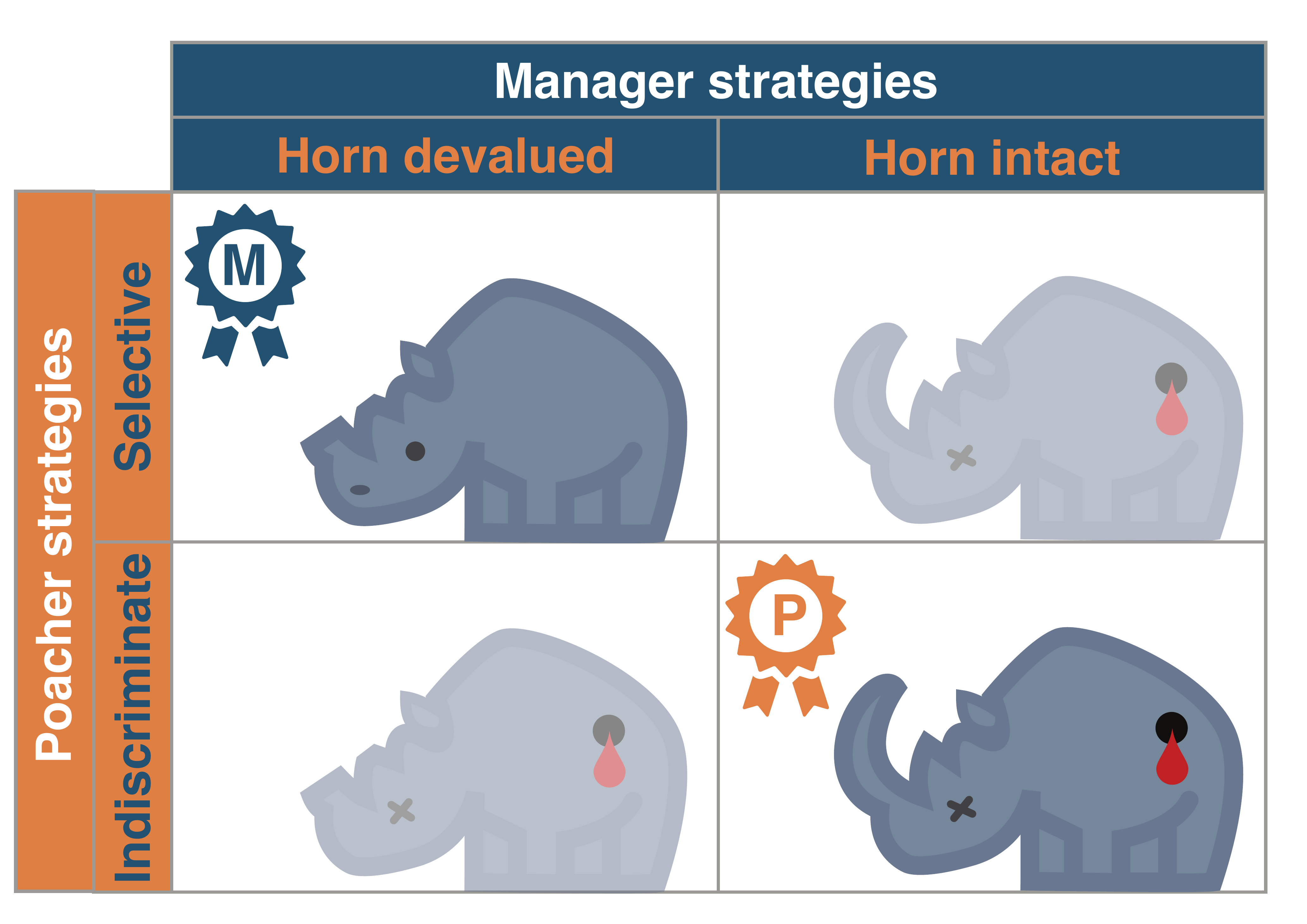}
    \caption{\label{fig:RhinoPic} The game between rhino manager and rhino
    poachers. The system settles to one of two equilibriums, either devaluing is effective or not.}
\end{figure}

At the extremes, we could consider the game as one of opportunistic exploitation~\cite{branch2013opportunistic}.
That is, consider intact rhinos and devalued rhinos as two species, where one is more valuable than
the other. Opportunistic exploitation advances upon the theory of anthropogenic Allee
effect to consider two species which are exploited together. Specifically, when a
highly valued species becomes rarer, a secondary, less valuable species is then targeted.
As with opportunistic exploitation on a larger scale, rhino managers need to account for
the multispecies system.

In this manuscript, we explore the population dynamic effects associated with
the interactions described by~\cite{Lee}. More specifically, the interaction between
poachers. In a population full of indiscriminate poachers is there
a benefit to a single poacher becoming selective or vice versa? This notion
is explored here using evolutionary game theory~\cite{Smith}. The
game is not that of two players anymore (manager and poacher) but now the players
are an infinite population of poachers. This allows for the interaction between
poachers over multiple plays of the game to be explored with the rhino manager
being the one that creates the conditions of the population.

Note that poachers are, in practice finite, and each has individual factors that will
affect a poacher's behaviour. An infinite population model corresponds to either
an asymptotic generalisation or overall descriptive behaviour.

In evolutionary game theory, we assume infinite populations and in our
model this is represented by \((x_1, x_2)\) with \(x_1\) being the proportion
of the population using a strategy of the first type and \(x_2\) of the second. We
assume there are utility functions \(u_1\) and \(u_2\) that map the population
to a fitness for each strategy, given by,

\[ u_1(x_1, x_2)  \text{ and } u_2(x_1, x_2).\]

In evolutionary game theory these utilities are used to dictate the evolution of
the population over time, according to the following replicator equations,

\begin{eqnarray}
    \label{eqn:u_differential_eq}
    \left\{
    \begin{array}{cl}
    \dfrac{dx_1}{dt}=x_1(u_1(x_1, x_2)-\phi),
    \\
    \\
    \dfrac{dx_2}{dt}= x_2(u_2(x_1, x_2)-\phi),
    \end{array} \right.
\end{eqnarray}
where \(\phi\) is the average fitness of the whole population~\cite{nowak2006evolutionary}.
In some settings these utilities are referred to as fitness and/or are mapped to
a further measure of fitness. This is not the case considered here (it is
assumed all evolutionary dynamics are considered by the utility measures).

Here, the overall
population is assumed to remain stable thus, \(x_1 + x_2 = 1 \) and
\begin{eqnarray}
    \dfrac{dx_1}{dt}  + \dfrac{dx_2}{dt} = 0 \Rightarrow x_1(u_1(x_1, x_2) - \phi)
     + x_2(u_2(x_1, x_2) - \phi)=0.
\end{eqnarray}

Recalling that \(x_1 + x_2 = 1\) the average fitness can be written as,

\begin{eqnarray}
\label{eqn:average_fitness}
    \phi=x_1u_1(x_1, x_2) + x_2u_2(x_1, x_2).
\end{eqnarray}

\noindent{} By substituting (\ref{eqn:average_fitness}) and \(x_2= 1 - x_1\)
in (\ref{eqn:u_differential_eq}),

\begin{eqnarray}
    \label{eqn:u_differential_simplified}
    \frac{dx_1}{dt}= x_1(1 - x_1)(u_1(x_1, x_2) - u_2(x_1, x_2)).
\end{eqnarray}

The equilibria of the differential equation (\ref{eqn:u_differential_simplified})
are given by, \(x_1=0\), \(x_1=1\), and \(0<x_1<1\) for \(u_1(x_1, x_2)=u_2(x_1, x_2)\).
These equilibria correspond to stability of the population: the differential
equation (\ref{eqn:u_differential_simplified}) no longer changes.

The notion of evolutionary stability can be checked only for these stable strategies.
For a stable strategy to be an Evolutionary Stable Strategy (ESS) it must remain
the best response even in a mutated population \((x_1, x_2)_\epsilon\). A mutated population
is the post entry population
where a small proportion \(\epsilon > 0\) starts deviating and adopts a different strategy.

In Section~\ref{section:the_model}, we determine expressions
for \(u_1, u_2\) that correspond to a population of wild rhino poachers and we
explore the stability of the equilibria identified in~\cite{Lee}. The results
contained in this paper are proven analytically, and more specifically it is
shown that:

\begin{itemize}
    \item In the presence of sufficient risk: a population of selective poachers
        is stable, meaning dehorning is a viable option.
    \item Full devaluation of all rhinos will lead to indiscriminate poachers.
\end{itemize}

\section{The Utility Model}\label{section:the_model}

As discussed briefly in Section~\ref{section:introduction}, a rhino poacher
can adopt two strategies, to either behave selectively
or indiscriminately. To calculate the utility for each strategy, the gain and cost
that poachers are exposed to must be taken into account. The poacher incurs a
loss from seeking a rhino, and the risk involved. The gain depends upon the value
of horn, the proportion of horn remaining after the manager has devalued the
rhino horn and the number of rhinos (devalued and not).

Let us first consider the gain to the poacher, where \(\theta\) is the amount of
horn taken. We assume rhino horn value is determined by weight only, a
reasonable assumption as rhino horn is sold in a grounded form~\cite{Saverhino}.
Clearly if the horn is intact, the amount of horn gained is \(\theta=1\) for
both the selective and the indiscriminate poacher.  If the rhino horn has been
devalued, and the poacher is selective, the amount of horn gained is
\(\theta=0\) as the poacher does not kill. However, if the poacher is behaving
indiscriminately, the proportion of value gained from the horn is \(\theta =
\theta_r\) (for some \(0<\theta_r<1\)). Therefore, the amount of horn gained in
the general case is

\begin{eqnarray}
    \label{eqn:theta}
    \theta(r, x) = x (1 - r) + (1 - x)(1-r+r\theta_r)
\end{eqnarray}

where \(r\) is the proportion of rhinos that have been devalued, and \(x\) is the
proportion of selective poachers and \(1-x\) is the proportion of indiscriminate
poachers. Note that since \(\theta_r, r, x  \in [0, 1]\), then
\(\theta(r, x) > 0\), that is, some horn will be taken. Standard supply and demand
arguments imply that the value
of rhino horn decreases as the quantity of horn available increases~\cite{mankiw2010}.
Thus at any given time the expected gain is

\begin{eqnarray}
    \label{eqn:individual_gain}
    H \theta(r, x)^{-\alpha},
\end{eqnarray}

where \(H\) is a scaling factor associated with the value of a full horn,
\(\alpha \geq 0\) is a constant that determines the precise relationship between
the quantity and value of the horn.  Fig.~\ref{fig:GainCurve}, verifies that the
gain curve corresponds to a demand curve: we see that as \(r\) increases, so that
the supply of horn decreases, the value is higher and vice versa. We have chosen
a simple function to model the demand (and thus gain) of rhino horn value,
relative to the proportion of rhinos devalued. However, demand for illegal
wildlife generally involves more factors than simply supply. Additional factors
include, but are not limited to, social stigma, tourism revenues, government
corruption, and rich countries being willing to pay to ensure species
existence~\cite{bulte1999economics, van2008protecting}.

\begin{figure}[!htbp]
\centering
\includegraphics[width=0.4\textwidth]{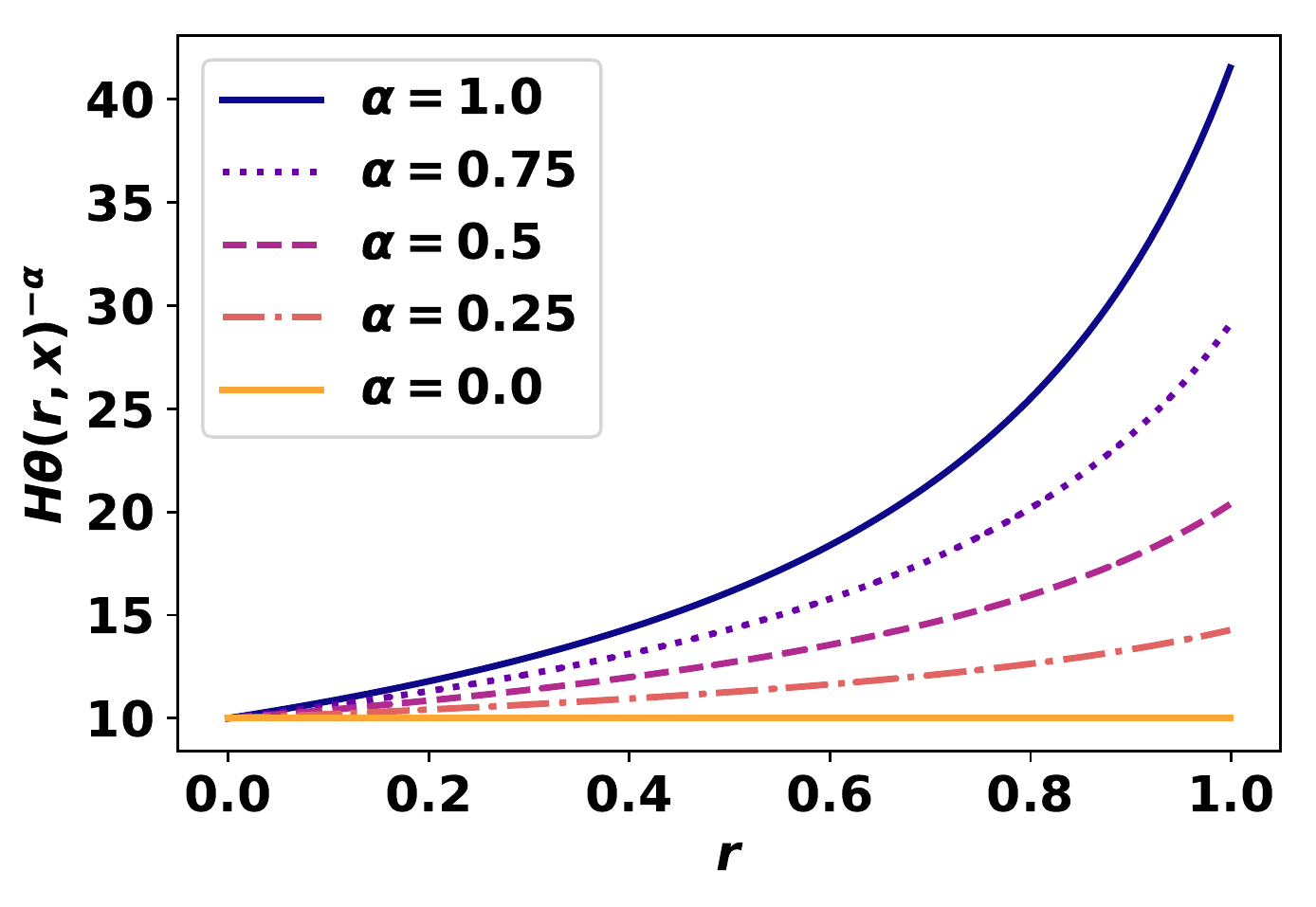}
\caption{\label{fig:GainCurve} \(H \theta(r, x) ^{- \alpha}\) for values
\(H = 10, \theta_r = 0.3\) and \(x = 0.2.\)}
\end{figure}

An individual interacts with the population which is uniquely determined by
\(x\), the proportion of selective poachers.
Therefore, the gain for a poacher in the population \(x\) is either

\begin{eqnarray}
    \label{eqn:gain}
    \left\{
    \begin{array}{cl}
    \theta(r, 1) H \theta(r, x)^{-\alpha} & \mbox{ selective poacher}
    \\
    \theta(r, 0) H \theta(r, x)^{-\alpha} & \mbox{ indiscriminate poacher}
    \end{array} \right.
\end{eqnarray}
depending on the chosen strategy of the individual.

Secondly we consider the costs incurred by the poacher. It is assumed that a given
poacher will spend sufficient time in the park to
obtain the equivalent of at least a single rhinoceros's horn. For selective poachers this
implies searching the park for a fully valued horn and for indiscriminate
poachers this implies either finding a fully valued horn or finding \(N_r\)
total rhinoceroses where \(N_r = \lceil \frac{1}{\theta_r} \rceil\).

Figure~\ref{fig:random_walk} shows a random walk that any given poacher will
follow in the park. Both types of poacher will exit the park as soon as they
encounter a fully valued rhino, which at every encounter is assumed to happen
with probability \(1 - r\). However, the indiscriminate poachers may also exit
the park if they encounter \(N_r\) devalued rhinos in a row.  Each step on the
random walk is assumed to last 1 time unit: during which a rhino is found.  To
capture the fact that indiscriminate poachers will spend a different amount of
time to selective poachers with each rhino the parameter \(\tau\) is introduced
which corresponds to the amount of time it takes to find and kill a rhino (thus
\(\tau\geq 1\)).

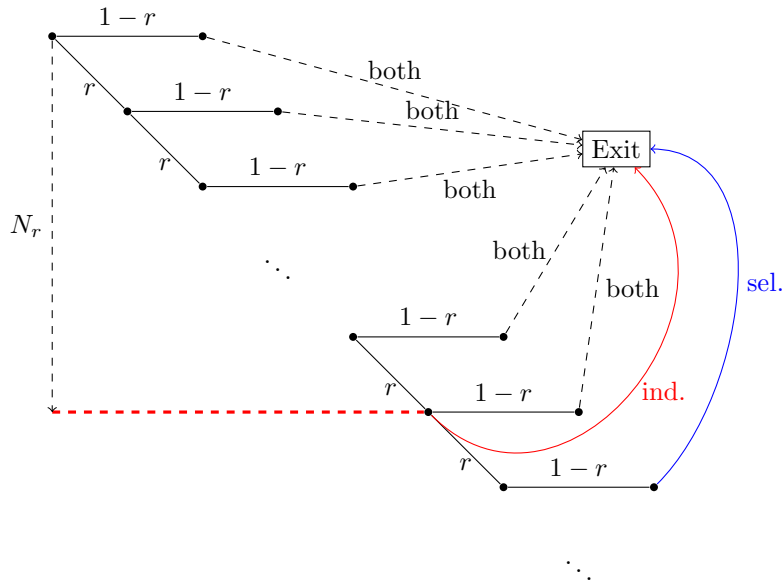
\begin{figure}[!htbp]
    \begin{center}
        \begin{tikzpicture}
            \node [circle,fill=black,inner sep=0pt,minimum size=3pt]
                (1) at (0, 0) {};
            \node [circle,fill=black,inner sep=0pt,minimum size=3pt]
                (1full) at ($(1) + (2, 0)$) {};
            \node [circle,fill=black,inner sep=0pt,minimum size=3pt]
                (2) at (1, -1) {};
            \node [circle,fill=black,inner sep=0pt,minimum size=3pt]
                (2full) at ($(2) + (2, 0)$) {};
            \node [circle,fill=black,inner sep=0pt,minimum size=3pt]
                (3) at (2, -2) {};
            \node [circle,fill=black,inner sep=0pt,minimum size=3pt]
                (3full) at ($(3) + (2, 0)$) {};
            \node (etc) at (3, -3) {\(\ddots\)};
            \node [circle,fill=black,inner sep=0pt,minimum size=3pt]
                (N-1) at (4, -4) {};
            \node [circle,fill=black,inner sep=0pt,minimum size=3pt]
                (N-1full) at ($(N-1) + (2, 0)$) {};
            \node [circle,fill=black,inner sep=0pt,minimum size=3pt]
                (N) at (5, -5) {};
            \node [circle,fill=black,inner sep=0pt,minimum size=3pt]
                (Nfull) at ($(N) + (2, 0)$) {};
            \node [circle,fill=black,inner sep=0pt,minimum size=3pt]
                (N+1) at (6, -6) {};
            \node [circle,fill=black,inner sep=0pt,minimum size=3pt]
                (N+1full) at ($(N+1) + (2, 0)$) {};
            \node (etc2) at (7, -7) {\(\ddots\)};

            \draw (1) -- node [below] {\(r\)} (2);
            \draw (1) -- node [above] {\(1 - r\)} (1full);
            \draw (2) -- node [below] {\(r\)} (3);
            \draw (2) -- node [above] {\(1 - r\)} (2full);
            \draw (3) -- node [above] {\(1 - r\)} (3full);
            \draw (N-1) -- node [below] {\(r\)} (N);
            \draw (N-1) --  node [above] {\(1 - r\)} (N-1full);
            \draw (N) -- node [below] {\(r\)} (N+1);
            \draw (N) --  node [above] {\(1 - r\)} (Nfull);
            \draw (N+1) --  node [above] {\(1 - r\)} (N+1full);

            \node [draw] (exit) at (7.5, -1.5) {Exit};
            \draw [dashed, ->] (1full) -- node [above] {both} (exit);
            \draw [dashed, ->] (2full) -- node [above] {both} (exit);
            \draw [dashed, ->] (3full) -- node [below] {both} (exit);
            \draw [dashed, ->] (N-1full) -- node [left] {both} (exit);
            \draw [dashed, ->] (Nfull) -- node [right] {both} (exit);
            \draw [blue, ->] (N+1full) edge [out=45, in=0] node [right] {sel.} (exit);

            \draw [red, dashed, very thick] ($(N) - (5, 0)$) -- (N);
            \draw [dashed, ->] (1) -- node [left] {\(N_r\)} (0, -5);
            \draw [red, ->] (N) edge[out=-45, in=-45, looseness=1.5] node [right] {ind.}
                (exit);
        \end{tikzpicture}
    \end{center}
    \caption{Illustrative random walk showing the points at which an
    indiscriminate or a selective poacher will leave the park.}

    \label{fig:random_walk}
\end{figure}

Using this, the expected time spent in the park \(T_1, T_2\) by poachers of both
types can be obtained:

For selective poachers:

\begin{equation}
    \begin{split}
    T_1 &= (1-r) \tau + r(1-r)(1+ \tau) + r^2(1-r)(2+ \tau) + \dots\\
        &= (1-r) \sum_{i=0}^{\infty}r^i(i + \tau)\\
        &= (1-r) \left ( \frac{1} {r} \sum_{i=0} ^ {\infty} i r ^ {(i+1)} + \tau \sum_{i=0} ^ {\infty} r ^ i \right)\\
        &= (1-r)\left(\frac{r}{(1-r)^2} + \frac{\tau}{1-r}\right) \quad \quad \quad \quad \text{ \small{using standard
formula for geometric series}}\\
    &= \frac{r+\tau(1-r)}{1-r}
    \end{split}
\end{equation}

For indiscriminate poachers:

\begin{equation}
    \begin{split}
    T_2 &= (1-r)\tau + r(1-r)2\tau + r^2(1-r)3\tau + \dots + r^{N_r -
    2}(1-r)(N_r - 1)\tau + r^{N_r - 1}N_r\tau\\
        &= (1-r)\tau\sum_{i=1}^{N_r - 1}ir^{i-1} + r^{N_r - 1}N_r\tau \\
        &= (1-r)\tau\left(\frac{1}{r \left(r - 1\right)^{2}} \left(N_{r} r
r^{N_{r}} - N_{r} r^{N_{r}} - r r^{N_{r}} + r\right)\right) + r^{N_r - 1}N_r\tau \\
        &= \frac{\tau(1 - r^{N_r})}{(1 - r)}
    \end{split}
\end{equation}

Figure~\ref{fig:indiscriminate_vs_selective_time} shows \(T_1\) and \(T_2\) for
varying values of \(r\) and \(\tau\) highlighting that \(\tau\) has a greater
effect on \(T_2\) than \(T_1\). Also, as \(r\) increases the overall time spent
in the park by both poachers increases and the value of \(\tau\) at which
\(T_1\) and \(T_2\) are equal increases.

\begin{figure}[!htbp]
    \begin{center}
        \includegraphics[width=\linewidth]{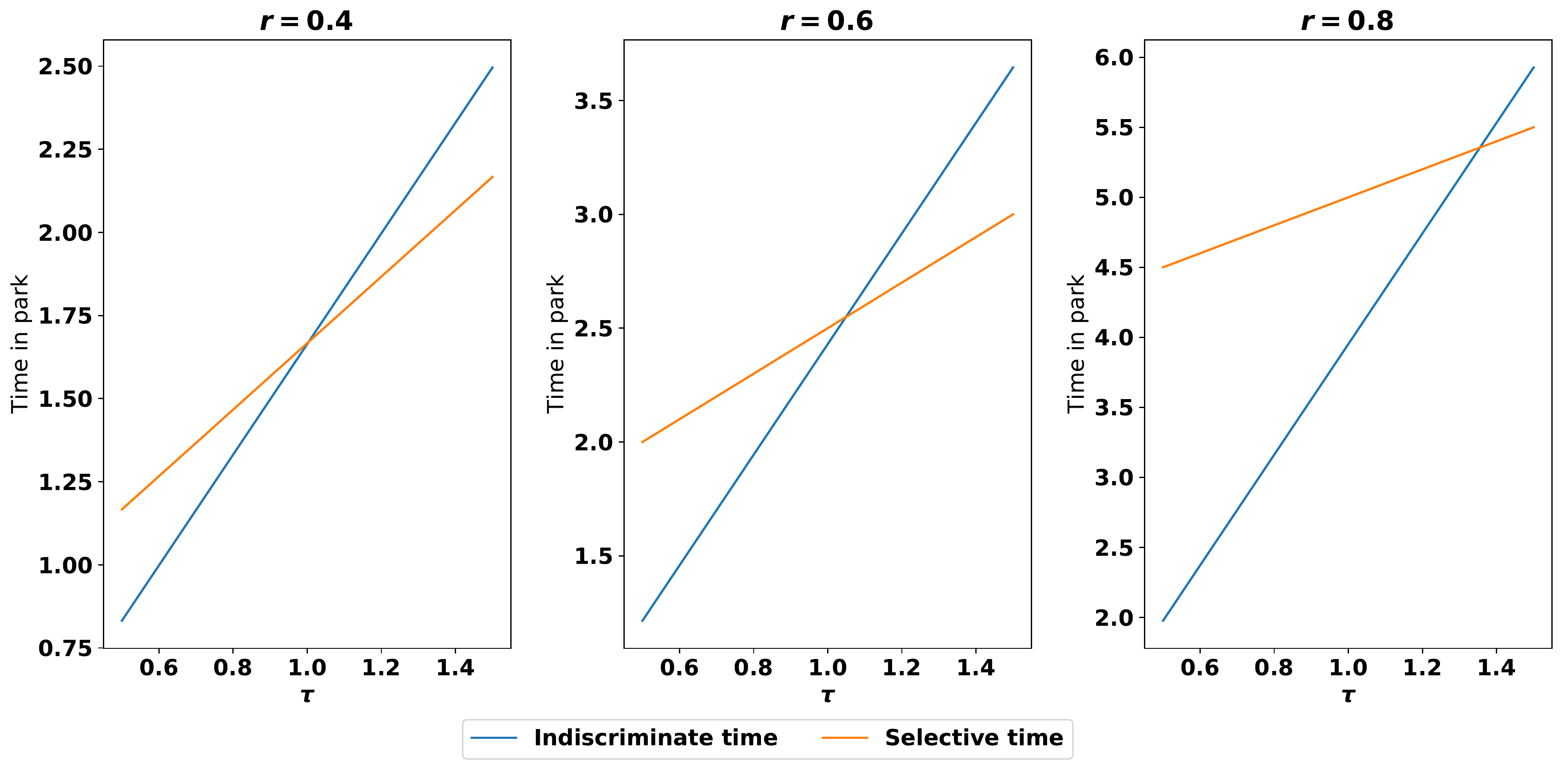}
        \caption{Expected time spent in the park for selective \(T_1\) and
        indiscriminate poachers \(T_2\) for \(\theta_r=.15\).}\label{fig:indiscriminate_vs_selective_time}
    \end{center}
\end{figure}

Additionally, the poachers are also exposed to a risk. The risk to the poacher is
directly related to the proportion of rhinos not devalued, \(1 - r\), since
the rhino manager can spend more on security if the cost of devaluing is low.
In real life this is not always the case. The cost of security can be extremely
high thus it cannot be guaranteed that much security will be added from the
saved money. However, our model assumes that there is a proportional and negative
relationship between the measures.

\begin{eqnarray}
    \label{eqn:risk}
    (1 - r)^{\beta},
\end{eqnarray}

where \(\beta \geq 0\) is a constant that determines the precise relationship
between the proportion of rhinos not devalued and the security on the grounds.
Therefore, at any given time the expected cost (due to the trade off between
security and devaluing) for a poacher is,

\begin{equation}
    \begin{split}
    \label{eqn:cost}
        FT_i(1 - r)^{\beta} \text{ for } i \in \{1, 2\}
    \end{split}
\end{equation}
where \(F\) is a constants that determines the precise relationship. Fig.
\ref{fig:CostCurves} verifies the decreasing relationship between \(r\) and the
cost. Notice that the cost to indiscriminate poachers remains fairly consistent,
irrespective of the proportion of devalued rhinos, until this proportion gets high.
Whereas the cost to selective poachers is more sensitive to the proportion of
rhinos devalued, especially when the time to kill a rhino is large.

\begin{figure}[!htbp]
    \begin{center}
        \includegraphics[width=\linewidth]{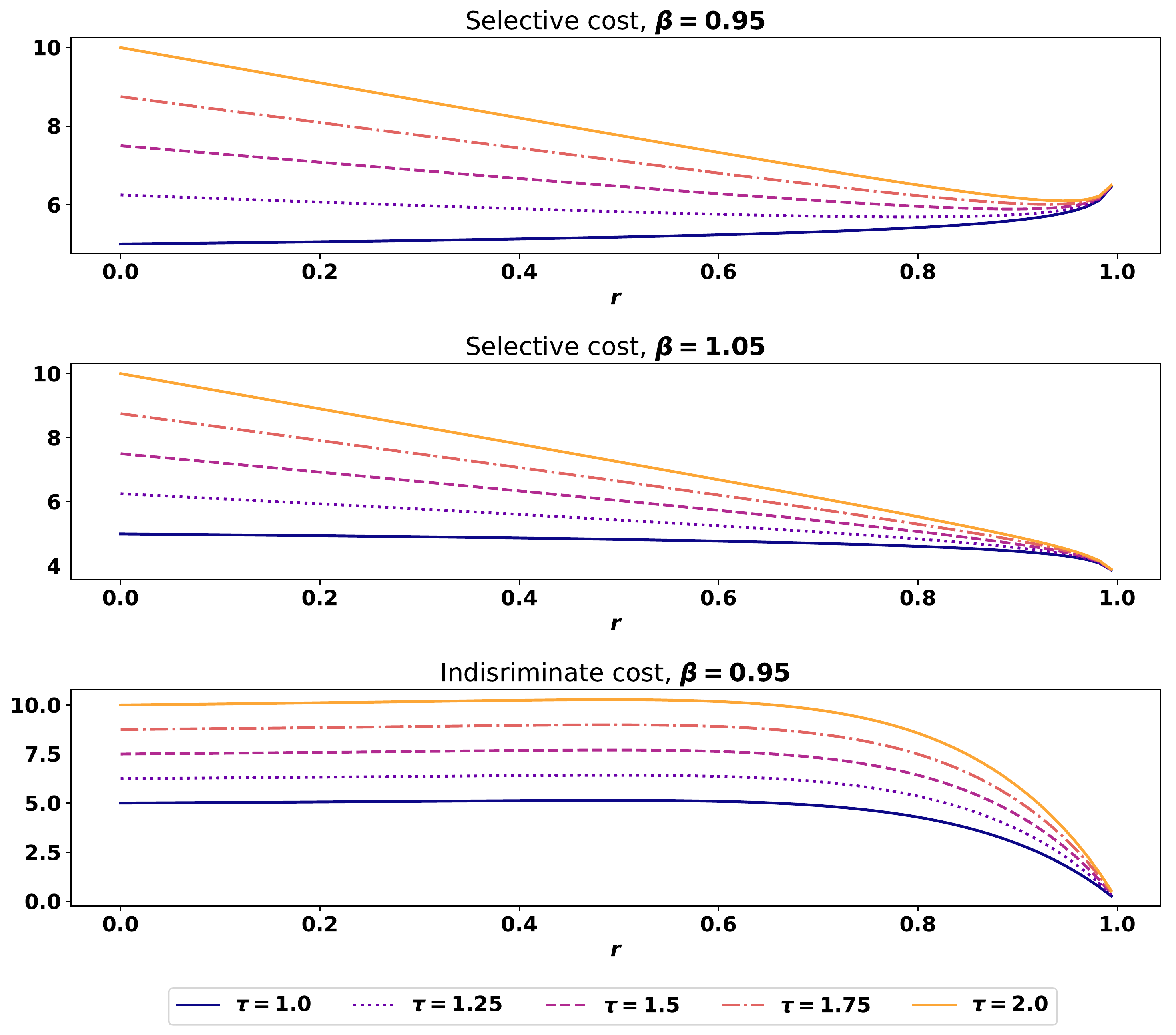}
        \caption{Costs associated to both poachers for \(F=5\) and varying
        values of \(r\) and \(\tau\).}\label{fig:CostCurves}
    \end{center}
\end{figure}

One final consideration given to the utility model is the incorporation of a
disincentive to indiscriminate poachers. Numerous interpretations can be
incorporated with this:

\begin{itemize}
    \item more severe punishment for indiscriminate killing of rhinos;
    \item educational interventions that highlight the negative aspects of
        indiscriminate killing;
    \item the possibility of a better alternative being offered to selective
        poachers.
\end{itemize}

This will be captured by a constant \(\Gamma\).

\noindent Combining~(\ref{eqn:gain}) and (\ref{eqn:cost}) gives the utility
functions for selective poachers, \(u_1(x)\), and indiscriminate poachers,
\(u_2(x)\),

\begin{equation}
    \begin{split}
\label{eqn:USchi}
u_1(x) = \theta(r,1) H \theta(r,x)^{-\alpha}
           - (r+ \tau (1-r))F (1-r)^{\beta - 1} ,
    \end{split}
\end{equation}

\begin{equation}
    \begin{split}
\label{eqn:UnotSchi}
u_2(x) = \theta(r,0) H \theta(r,x)^{-\alpha} - \tau (1 - r^{N_r})F(1-r)^{\beta-1}
- \Gamma
\end{split}
\end{equation}

Given a specific individual, let \(s\) denote the probability of them behaving
selectively.
Thus the general utility function for an individual poacher in the population with
a proportion of \(0 \leq x \leq 1\) selective poachers is

\begin{equation}
\label{eqn:utility}
u(s, x) = s u_1(x) +(1 - s) u_2(x).
\end{equation}
Substituting~(\ref{eqn:USchi}) and~(\ref{eqn:UnotSchi}) into~(\ref{eqn:utility})
and using (\ref{eqn:theta}) gives,

\begin{equation}
    \begin{split}
\label{eqn:tutility2}
u(s, x) = H (\theta_r r(1-s) - r + 1)\theta(r,x)^{-\alpha} - F\left(sr+s\tau(1-r)+
(1-s)\tau(1-r^{N_r})\right)(1-r)^{\beta-1} -(1-s)\Gamma
    \end{split}
\end{equation}

Figure~\ref{fig:evolution_of_system} shows the evolution of the system over time
for a variety of initial populations and parameters.
This is done using numerical integration implemented in~\cite{scipy}.
All the source code used for this work has been written in a sustainable manner: it is
open source (\url{https://github.com/Nikoleta-v3/Evolutionary-game-theoretic-Model-of-Rhino-poaching/})
and tested which ensures the validity of the results. The source code
has also been properly archived and can be found at \cite{Evorepo2018}.

A summary of all the parameters and their meanings is given by Table~\ref{tbl:parameters}.

\begin{table}[!hbtp]
    \begin{center}
        \begin{tabular}{cl}
            \toprule
            Parameter      & Interpretation                \\
            \midrule
            $\theta_r$     & the proportion of value gained from a devalued horn \\
            $r$            & the proportion of rhinos that have been devalued\\
            $H$            & a scaling factor associated with the value of a full horn\\
            $\alpha$       & the relationship between the quantity and value of a horn \\
            $F$            & the cost of retrieving a horn\\
            $\beta$        & the relationship between the proportion of devalued rhinos and security\\
            $\tau$         & the time it takes to find and kill a rhino\\
            $\Gamma$       & a disincentive only applied to indiscriminate poachers \\
            \bottomrule
        \end{tabular}
    \end{center}
    \caption{A summary of the parameters used.}\label{tbl:parameters}
\end{table}

\begin{figure}[!htbp]
    \includegraphics[width=\textwidth]{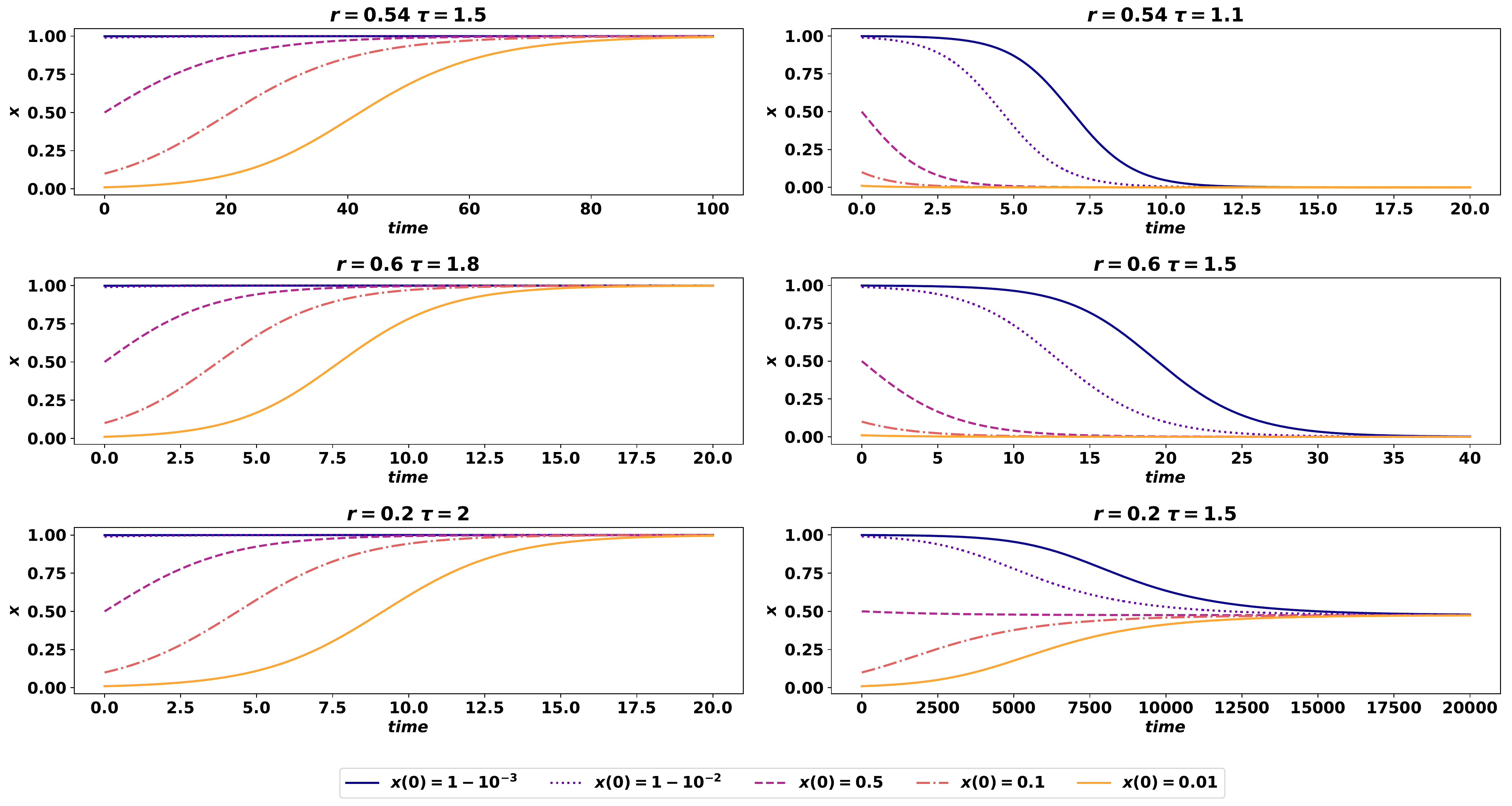}
    \caption{\label{fig:evolution_of_system} The change of the population over
    time with different starting populations. For \(F=5, H=50,
    \alpha=2, \beta=.99, \tau=1.5, \theta_r=0.01, \Gamma=0\).}
\end{figure}

In Figure~\ref{fig:evolution_of_system} the left column shows parameters sets
for which a selective population is stable (\(x=1\). The figures on the right
correspond to a decrease in \(\tau\) which decreases the risk associated with
acting indiscriminately: in these cases a population of selective poachers is
unstable. In one case (\(H=192\)) a mixed population is stable: some poachers
will continue to act selectively.

In Section~\ref{section:evolutionary_stability}, these observations will be
confirmed theoretically.

\section{Evolutionary Stability}\label{section:evolutionary_stability}

By definition, for a strategy to be an ESS it must first be a best response to an
environment where the entire population is playing the same strategy.
In our model there are three possible stable distributions based on the
equilibria of equation (\ref{eqn:u_differential_simplified}):

\begin{itemize}
    \item all poachers are selective;
    \item all poachers are indiscriminate;
    \item mixed population of selective and indiscriminate poachers.
\end{itemize}

An ESS corresponds to asymptotic behaviour near the equilibria of
(\ref{eqn:u_differential_simplified}), this correspond to the concept of
Lyapunov stability~\cite{lyapunov1992general}.

For simplicity, denote the right hand side of
(\ref{eqn:u_differential_simplified}) as \(f\).
In this setting, when \(x\) is near to some equilibria \(x^*\) so that
\(f(x^*)=0\) then the evolutionary game can be linearized (using standard Taylor
Series expansion) as:

\begin{equation}
    \begin{split}
    \frac{d(x^* + \epsilon)}{dt} = J(x^*)\epsilon
    \end{split}
\end{equation}

where:

\begin{equation}
    \begin{split}
    J(a) = \left.\frac{df}{dx}\right|_{x=a}
    \end{split}
\end{equation}

This gives a standard approach for evaluating equilibria of the underlying game.
For a given equilibria \(x^*\), \(J(x^*)<0\) if and only if \(x^*\) is an ESS\@.

Using equations (\ref{eqn:USchi}) and (\ref{eqn:UnotSchi}):

\begin{equation}
    \begin{split}
    J(a) = \frac{1}{\left(r - 1\right) \left(- a r \theta_{r} + r \theta_{r} - r
+ 1\right)^{\alpha + 1}} \left(J_1 - J_2 \right) + \Gamma(1 - 2a)
    \end{split}
\end{equation}

where:

\begin{equation*}
    \begin{split}
    J_1 & = F \left(- r + 1\right)^{\beta} \left(- a r \theta_{r} + r \theta_{r}
    - r + 1\right)^{\alpha + 1} \left(2 a r t - 2 a r - 2 a
r^{\lceil{\frac{1}{\theta_{r}}}\rceil} t - r \tau + r +
r^{\lceil{\frac{1}{\theta_{r}}}\rceil} \tau\right)\\
    J_2 & = H a \alpha r^{2} \theta_{r}^{2} \left(- a + 1\right) \left(r - 1\right) + H r \theta_{r} \left(2 a - 1\right) \left(r - 1\right) \left(a r \theta_{r} - r \theta_{r} + r - 1\right)
    \end{split}
\end{equation*}

\begin{theorem}\label{theorem:selective}
Using the utility model described in Section~\ref{section:the_model},
a population of selective poachers is stable if and only if:

\begin{equation}
    \begin{split}\label{eq:selective_stability}
    \tau > \frac{1}{1 - r^{\lceil{\frac{1}{\theta_{r}}}\rceil - 1}}
    \frac{F + H \theta_{r} (1 - r) ^{1- \alpha -\beta}-\frac{\Gamma}{r(1-r)^{1-\beta}}}{F}
\end{split}
\end{equation}
\end{theorem}

\begin{proof}
Direct substitution gives:

\begin{equation*}
    \begin{split}
    J(1) &= \frac{1}{\left(- r + 1\right)^{\alpha + 1} \left(r - 1\right)}
                \left(
                    F \left(- r + 1\right)^{\beta}
                    \left(- r + 1\right)^{\alpha + 1}
                    \left(
                        r \tau -
                        r -
                        r^{\lceil{\frac{1}{\theta_{r}}}\rceil} \tau
                    \right)
                    - H r
                    \theta_{r}
                    \left(r - 1\right)^{2}
                \right) - \Gamma\\
		&= \left(F
                 \left(1 - r\right)^{\beta - 1}
                 \left(
                     r -
                     \tau(r - r^{\lceil{\frac{1}{\theta_{r}}}\rceil} )
                 \right) +
                 H r \theta_{r} \left(1 - r\right)^{-\alpha}
           \right)
           - \Gamma
        \end{split}
    \end{equation*}

The required condition is \(J(1)<0\):

\begin{equation*}
    \begin{split}
    F
    \left(1 - r\right)^{\beta - 1}r +
    H r \theta_{r}
    \left(1 - r\right)^{-\alpha}
    - \Gamma
            & < F \left(1 - r\right)^{\beta - 1}
    \tau(r - r^{\lceil{\frac{1}{\theta_{r}}}\rceil} )\\
    \frac{r}{r - r^{\lceil{\frac{1}{\theta_{r}}}\rceil}}
    \frac{F + H \theta_{r} \left(1 - r\right)^{1-\beta-\alpha} -
    \frac{\Gamma}{r(1-r)^{1-\beta}}}{F} & < \tau
\end{split}
\end{equation*}

which gives the required result.
\end{proof}

Note that the limit of the right hand side of equation (\ref{eq:selective_stability})
tends to infinity as \(r\to1^-\). This means that devaluing all rhinos is not
a valid approach.

Furthermore, we see that the equilibria with poachers acting selectively,
predicted in~\cite{Lee} can in fact be obtained in specific settings.

Note that similar theoretic results have been obtained about the evolutionary
stability of indiscriminate poachers but these have been omitted for the sake of
clarity.

In this section we have analytically studied the stability of all the possible
equilibria. We have proven that all potential equilibria are possible.  All of
these theoretic results have been verified empirically, and the data for this
has been archived at~\cite{Glynatsi2017}.
Figure~\ref{fig:convergence-over-r} shows a number of scenarios where \(F=5\),
\(\beta=0.99\), \(\theta_r = 0.05\) and, unless varied as stated on the x-axis:

\begin{itemize}
    \item Scenario 1: \(H=50\)  \(r=0.45\) \(\alpha=2\), \(\tau=2\), \(\Gamma=0\)
\item Scenario 2: \(H=50\)  \(r=0.4\) \(\alpha=2.5\), \(\tau=1.8\), \(\Gamma=0\)
\item Scenario 3: \(H=25\)  \(r=0.45\) \(\alpha=2\), \(\tau=2\), \(\Gamma=0\)
\item Scenario 4: \(H=25\)  \(r=0.4\) \(\alpha=2.5\), \(\tau=1.8\), \(\Gamma=0\)
\item Scenario 5: \(H=25\)  \(r=0.99\) \(\alpha=2\), \(\tau=2\), \(\Gamma=4\)
\item Scenario 6: \(H=25\)  \(r=0.99\) \(\alpha=2.5\), \(\tau=1.8\), \(\Gamma=4\)

\end{itemize}

\begin{figure}[!htbp]
    \includegraphics[width=\textwidth]{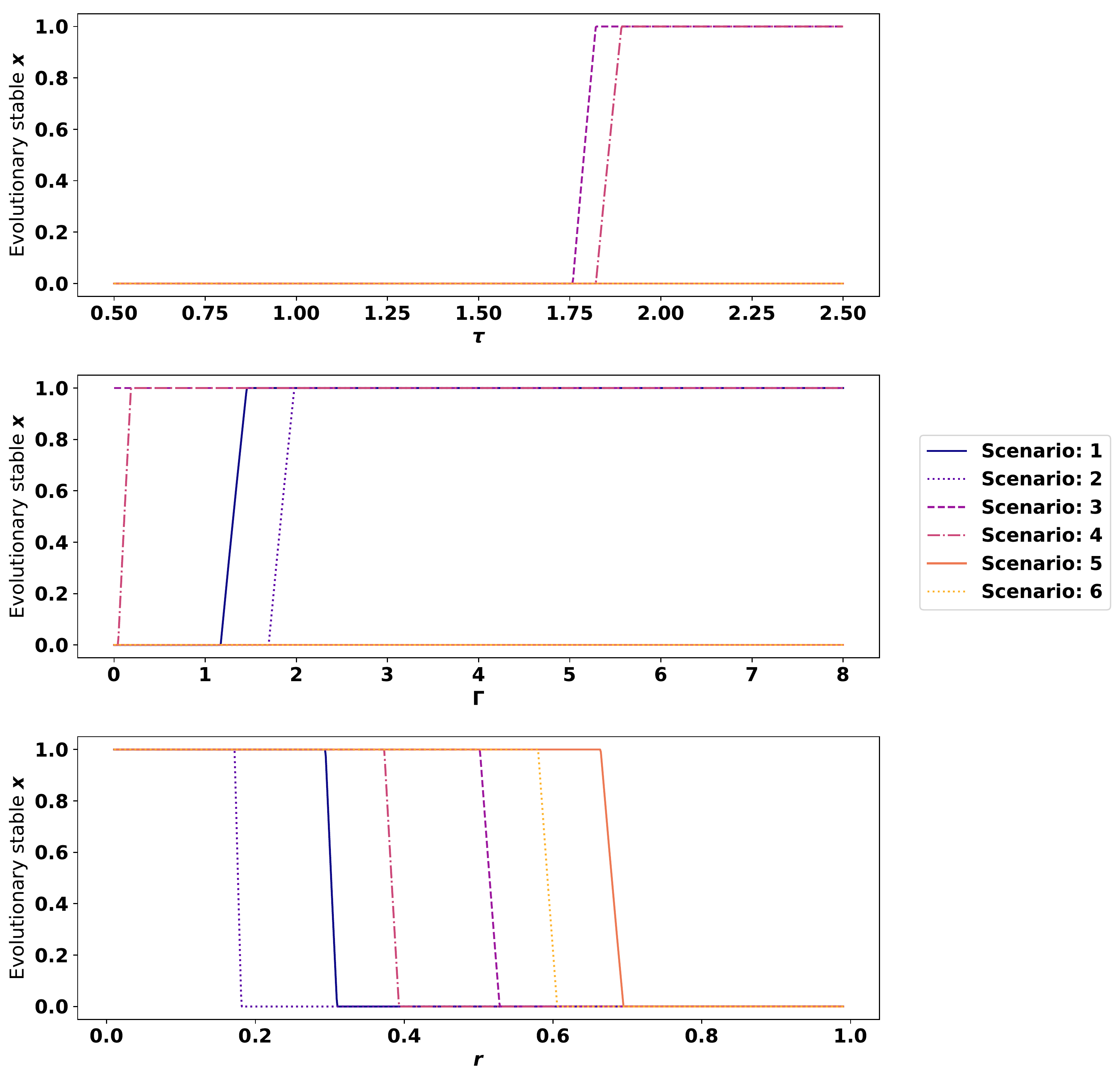}
    \caption{Evolutionary stable populations for varying values of \(\tau, r,
    \Gamma\) for 6 difference scenarios.}
    \label{fig:convergence-over-r}
\end{figure}

The first plot in Figure~\ref{fig:convergence-over-r} shows that when the value
of a full horn is low, and there is no disincentive factor (Scenarios 3 and 4),
the poacher strategy can be influenced by the time taken to kill a rhino such
that a long time can push the poacher to behave selectively (as there is
increased risk associated with acting indiscriminately). The second plot, shows
that the disincentive factor has most influence when the value of a full horn is
high (Scenarios 1 and 2). Otherwise, poachers will generally be indiscriminate
if \(r\) is large (Scenarios 5 and 6) or selective otherwise (Scenarios 1, 2, 3
and 4).  Most importantly, the third plot
confirms Theorem~\ref{theorem:selective}. A high value of \(r\)
forces the population to become indiscriminate even with a high disincentive.
Moreover, for all scenarios a value of \(r\) does exist for which a selective
population will subsist.

This confirms that devaluing alone is not a solution and in fact can
potentially have averse consequences: combinations of devaluing and
education (creating a disincentive) is needed.

\section{Discussion and Conclusions}
\label{section:discussion}

In this work the dynamics of a selective population were explored. It was shown
that given sufficient risk associated with killing a rhino, it would be possible
for a selective population of poachers to subsist.

We have developed a game theoretic model which examines the specific question
for rhino managers: how to deter poachers by devaluing horns?  One of the main
conclusions of the work presented here is that if there is sufficient risk
associated with indiscriminate behaviour then a population of selective poachers
can be stable. The model also incorporates wider factors in a general manner such as
a disincentive factor. The disincentive factor may be an increase in the
monetary fine for poachers. In fact~\cite{di2015identification}, who identify
the most important contributors to the number of rhinos illegally killed in
South Africa (between 1900 and 2013), found that increasing the monetary fine
has a more significant effect than increasing the years in prison. However, the
disincentive factor may also include wider influences, such as engaging the
rural communities that neighbour wildlife~\cite{Duan2013}, or decreasing the
cost of living with wildlife, and supporting a livelihood that is not related to
poaching. Zooming out further, it could include global issues such as an
increase in ecotourism, which would provide a sustainable income for the
community.

Another opportunity for wider factors, such as global issues, to be included in
the model is via the supply and demand function. For example,~\cite{di2015identification}
show that one of the three most important contributors to the number of rhinos illegally
killed was the GDP in Far East Asia, where the demand for rhino horn is at its greatest.
This finding supports~\cite{lawson2014global} call for improved law enforcement and
demand reduction in the Far East.

Note that the proportion of devalued rhinos \(r\) is continuous over \([0, 1]\)
in the model. However, standard practice of a given park manager in almost all
cases is to either devalue all the animals in a defined enclosed area, or none
at all. This is thought to be because partial devaluing tends to disturb
rhino social structures. Our results indicate that devaluing all rhinos will
only decrease rhino poaching if potential poachers have a viable alternative
(even in the case of a large disincentive).

The debate about the effectiveness of devaluation for preventing poachers and,
is extensive and ongoing. This model answers
one aspect of the topic, but larger questions remain. There are many drivers to
account for, many of which are included in a systems dynamics model presented
in~\cite{crookes2016trading} which captures the five most important factors:
rhino abundance, rhino demand, a price model, an income model and a supply
model. Using the optimal dehorning model of~\cite{Milner1992}, the
model~\cite{crookes2016trading} finds that poachers behaving
indiscriminately will always prevail, which indicates that the risk associated
with indiscriminate behaviour might not have been captured fully.

Following discussions with environmental specialists it is clear that devaluing
is empirically thought to be one of the best responses to poaching. This
indicates that whilst of theoretic and potential macroeconomic interest, the
modelling approach investigated in this work has potential for further work. For
example, a detailed study of two neighbouring parks with differing policies
could be studied using a game theoretic model, this would require an
understanding of the travel times which can be very large and have a non
negligible effect. Another interesting study would
be to introduce a third strategy available to poachers: this would represent the
possibility of not poaching (perhaps finding another source of income) and/or
leaving the current environment to poach elsewhere. Finally, the specific rhino
population could also be modelled using similar techniques and incorporated in
the supply and demand model.

\section*{Authors' contributions}

All authors conceived the ideas and designed the methodology. NG and VK developed the
source code needed for the numerical experiments and generating the data. All authors
contributed critically to the drafts and gave final approval for publication.

\section*{Acknowledgements}

This work was performed using the computational facilities of the Advanced
Research Computing @ Cardiff (ARCCA) Division, Cardiff University.

A variety of software libraries have been used in this work:

\begin{itemize}
    \item The Scipy library for various algorithms~\cite{scipy}.
    \item The Matplotlib library for visualisation~\cite{hunter2007matplotlib}.
    \item The SymPy library for symbolic mathematics~\cite{sympy}.
    \item The Numpy library for data manipulation~\cite{walt2011numpy}.
\end{itemize}

We would like to express our appreciation to David Roberts and Michael 't
Sas-Rolfes for their valuable and constructive suggestions for this research.
We would also like to thank
Dr. Jonathan Gillard for his careful reading and double checking of the algebra in the paper,
and Anna Huzar for Figure~\ref{fig:RhinoPic}.

Finally, we would like to thank the referees whose comments on the first version 
of this manuscript greately improved the undertaken research.

\section*{Data Accessibility}

The data generated for this work have been archived and are available
online~\cite{Glynatsi2017}.

\bibliographystyle{plain}
\bibliography{bibliography.bib}

\end{document}